\documentclass[%
 reprint,
 amsmath,amssymb,
 aps,
]{revtex4-1}

\usepackage{float}

\usepackage{amsthm}
\usepackage{graphicx}
\usepackage{dcolumn}
\usepackage{bm}
\usepackage{float}
\usepackage{booktabs} 
\usepackage{amsmath,amsthm} 
\usepackage{setspace} 
\usepackage[sort&compress]{natbib}

\usepackage[top=2cm, bottom=2cm, left=2cm, right=2cm]{geometry}
\usepackage{algorithm}
\usepackage{algorithmicx}
\usepackage{algpseudocode}
\usepackage{amsmath}

\floatname{algorithm}{Algorithm}

\algrenewcommand{\algorithmiccomment}[1]{\hskip3em//#1}

\usepackage{color}

\newcommand{\jiaqing}[1]{}

\makeatletter  

\newcommand{\Rmnum}[1]{\expandafter\@slowromancap\romannumeral #1@}
\makeatother

\begin{document}

\title{Quantum-to-quantum Bernoulli factory problem}

\author{Jiaqing Jiang$^{1,2}$}
\email{jiangjiaqing17@mails.ucas.ac.cn}
\author{Jialin Zhang$^{1,2}$}
\email{zhangjialin@ict.ac.cn}
\author{Xiaoming Sun$^{1,2}$}
\email{sunxiaoming@ict.ac.cn}
\affiliation{
 $^{1}$1 CAS Key Lab of Network Data Science and Technology, Institute of Computing Technology, Chinese Academy of Sciences, 100190, Beijing, China.\\
 $^{2}$2 University of Chinese Academy of Sciences, Beijing, 100049, China.
}

\begin{abstract}

Given a coin with unknown bias $p\in [0,1]$, can we exactly simulate another coin with bias $f(p)$? The exact set of simulable functions has been well characterized 20 years ago. In this paper, we ask the quantum counterpart of this question: Given the quantum coin $|p\rangle=\sqrt{p}|0\rangle+\sqrt{1-p}|1\rangle$, can we exactly simulate another quantum coin $|f(p)\rangle=\sqrt{f(p)}|0\rangle+\sqrt{1-f(p)}|1\rangle$? We give the full characterization of simulable quantum state $k_0(p)|0\rangle+k_1(p)|1\rangle$ from quantum coin $|p\rangle=\sqrt{p}|0\rangle+\sqrt{1-p}|1\rangle$, and present an algorithm to transform it. 
Surprisingly, we show that simulable sets in the quantum-to-quantum case and classical-to-classical case have no inclusion relationship with each other.

\end{abstract}

\pacs{Valid PACS appear here}


\maketitle

\theoremstyle{remark}
\newtheorem{definition}{\indent Definition}
\newtheorem*{observation}{\indent Observation}
\newtheorem*{sufficiency}{\indent Sufficiency}
\newtheorem*{necessity}{\indent Necessity}
\newtheorem{lemma}{\indent Lemma}
\newtheorem{theorem}{\indent Theorem}
\newtheorem*{corollary}{\indent Corollary}
\newtheorem*{conjecture}{\indent Conjecture}
\newtheorem{statement}{\indent Statement}

\theoremstyle{definition}

\def\QEDclosed{\mbox{\rule[0pt]{1.3ex}{1.3ex}}}
\def\QED{\QEDclosed}
\def\proof{\indent{\em Proof}.}
\def\endproof{\hspace*{\fill}~\QED\par\endtrivlist\unskip}

\section{\label{sec:level1}Introduction}
The classical Bernoulli factory, proposed by Keane and O'Brien~\cite{classical} in 1992, considered the following problem: Given a biased classical coin, with unknown success probability (denoted by $p$), can we use it to simulate an arbitrary function $f(p)$? That means, can we construct an event with success probability which exactly equals $f(p)$? Note that $p$ is unknown and the event we choose should not depend on the value of $p$.
For example, $f(p)=p^2$ is simulable. We can flip the coin twice and the event with two heads has success probability $p^2$. Surprisingly, Keane and O'Brien~\cite{classical} proved that the probability amplification function $\min(2p,2-2p)$ for $p\in[0,1]$ is not simulable. In fact, they had determined the exact set of simulable functions. Roughly speaking, a function $f(p): [0,1]\rightarrow [0,1]$ is simulable if and only if it is continuous and the function value does not touch zero or one in $(0,1)$.

In 2015, Dale, Jennings and Rudolph~\cite{quantum} generalized the Bernoulli factory problem to the case of quantum input and classical output. The problem is, if we use quantum coins $|p\rangle=\sqrt{p}|0\rangle+\sqrt{1-p}|1\rangle$ instead of classical ones, what kind of function $f(p)$ is simulable? In other words, can we find an event with success probability $f(p)$? Obviously, if $f(p)$ is simulable in the classical-to-classical case, it is also simulable in the quantum-to-classical case. We can simply measure the coin $|p\rangle$ in computational bases and thus regard $|p\rangle$ as a classical coin. Furthermore, we can use unitary transformations in the simulation process and obtain some events which cannot be simulated in the classical case. For example, if we are given a coin $|p\rangle$, we can do a unitary transformation $U_a=\sqrt{a}|0\rangle\langle0|+\sqrt{1-a}|0\rangle\langle1|+\sqrt{1-a}|1\rangle\langle0|-\sqrt{a}|1\rangle\langle1|$ on $|p\rangle$ and then measure it. The probability of getting result $|1\rangle$ will be $f_a(p)=(\sqrt{p(1-a)}-\sqrt{a(1-p)})^2$ while $f_a(p)$ cannot be simulated in the classical case when $a\in(0,1)$. Dale \emph{et al.} characterized the exact set of simulable functions in their quantum-to-classical case, which is strictly larger than the classical-to-classical case. Specially, $\min(2p,2-2p)$ is simulable in the quantum-to-classical case. In 2016, Yuan~\emph{et al.} implemented experiments to show this quantum advantage~\cite{experiment}.

So far, the Bernoulli factory in classical-to-classical and quantum-to-classical cases has been solved. It is natural to ask how about the quantum-to-quantum case. As Dale \emph{et al.}~\cite{quantum} questioned: Given an unbounded number of quantum coins  $|p\rangle$, what kind of quantum coins $|f(p)\rangle=\sqrt{f(p)}|0\rangle+\sqrt{1-f(p)}|1\rangle$ can we get through unitary transformation and measurement?

In this paper, we further generalize the problem by allowing complex amplitudes. With an unbounded number of quantum coins $|p\rangle$, what kind of quantum states $k_0(p)|0\rangle+ k_1(p)|1\rangle$ where $|k_0(p)|^2+|k_1(p)|^2=1$ can we get through unitary transformation and measurement?
We answer the question by providing the exact characterization of the simulable states.
The characterization of $|f(p)\rangle$, in other words, $(k_0(p),k_1(p))= (\sqrt{f(p)},\sqrt{1-f(p)})$, will be given as a corollary to the main theorem.

The main difference between our work and the previous work is the quantum output. In both classical-to-classical and quantum-to-classical cases, they try to find an event with success probability $f(p)$ exactly. But, we aim to transform a quantum state $|p\rangle$ to a new state $|f(p)\rangle$, or more generally $k_0(p)|0\rangle+k_1(p)|1\rangle$, through unitary transformations and measurement, where $p$ is unknown. By simply measuring the constructed state in the computational basis, we can get classical output $f(p)=|k_0(p)|^2$ from quantum output. Thus, intuitively, the set of the probability $|k_0(p)|^2$, related to simulable tuple $(k_0(p),k_1(p))$,  in the quantum-to-quantum case is the subset of the simulable functions in the quantum-to-classical case. However, since our construction process contains quantum measurement, we will transform $|p\rangle$ to a new state only with a nonzero success probability. This means, the event ``measuring the constructed state in computational basis'' might not be a legal probability event in the quantum-to-classical case. Therefore, this intuition does not directly lead to such a result. In our paper, we use our characterization of simulable tuple in the quantum-to-quantum case to show that the classical output related to the simulable tuple is indeed the subset of that in the quantum-to-classical case.

However, although our result is a strict subset of Dale \emph{et al.}'s, there are obvious some advantages for our quantum output. The quantum state $|f(p)\rangle$ itself is much more useful in a large quantum system than an event with success probability $f(p)$. For example, consider two states $\frac{1}{\sqrt{p^2+1}}(p|0\rangle+|1\rangle)$ and $\frac{1}{\sqrt{q^2+1}}(q|0\rangle+|1\rangle)$, with $p$ and $q$ unknown. We may wonder if we can construct state $\frac{1}{\sqrt{(f(p)-q)^2+1}}( (f(p)-q)|0\rangle+|1\rangle)$. Using technique in this paper, we know the answer is yes when $|f(p)\rangle$ is constructible and the final state may help us judge whether $f(p)=q$ or not. Our result deepens our understanding of the power of quantum state manipulation.

Besides, our work also helps us find more efficient simulating methods in the quantum-to-classical case. In \cite{quantum}, Dale \emph{et al.} implemented a special coin and used it to simulate $f(p)$. It is natural to consider if there exists another coin which leads to better performance in the average case. Our results provide the basic set for the searching. On the other hand, when $(k_0(p),k_1(p))$ is simulable, our algorithm can simulate $f(p)=|k_0(p)|^2$ directly while the method in \cite{quantum} needs to do quantum operations as well as design experiments for classical simulation.

For those reasons, we are curious about the quantum-to-quantum case.
In this paper, we give a succinct condition to describe the exact set of simulable tuple $(k_0(p), k_1(p))$, where $|k_0(p)|^2+|k_1(p)|^2=1$. The main observation is that the ratio $\frac{k_0(p)}{k_1(p)}$ has an intrinsic field structure. The characterization of $(\sqrt{f(p)},\sqrt{1-f(p)})$, where $f(p)$ is a real function of $p$, is a corollary of the main result. Surprisingly, we find that $(p, \sqrt{1-p^2})$ is not simulable while $(e^{i\theta(p)}p, \sqrt{1-p^2})$ is simulable for some $\theta(p)$.

One key point in our paper is that $p$ is unknown, and the quantum operations we take should not depend on $p$. If we know $p$, we can calculate $(k_0(p),k_1(p))$ and use unitary transformations~\cite{Nilson} to transform $|0\rangle$ to $k_0(p)|0\rangle+k_1(p)|1\rangle$. When $p$ is unknown, we may first estimate the value of $p$~\cite{est1,est2,est3} and approximately construct the target state with arbitrary accuracy. Different from approximation, in this paper, our aim is to exactly construct $k_0(p)|0\rangle+k_1(p)|1\rangle$. Besides approximate transformation, there is a similar work called quantum sampling: given classical distribution $(\pi_i)_i$, the aim is to construct state $\sum\sqrt{\pi_i}|i\rangle$ efficiently \cite{q-sam1,q-sam2,q-sam3,q-sam4}. In quantum sampling, $(\pi_i)_i$ is known and the desired result is a muti-qubit state, while our work considers one qubit with unknown amplitude $(k_0(p), k_1(p))$.

The structure of this paper is as follows. In Sec. $\uppercase\expandafter{\romannumeral2}$, we give some formal definitions. In Sec. $\uppercase\expandafter{\romannumeral3}$, we provide our main theorem and the proof, as well as an algorithm to implement a given simulable state\jiaqing{or $|f(p)\rangle$ or tuple?}. After that, we show some examples of simulable states. We then compare the simulable sets in different cases. Finally we give a summary in Sec. $\uppercase\expandafter{\romannumeral4}$.

\section{\label{sec:level1}Preliminaries}

We give the formal definition of simulable tuple in the quantum input and quantum output case.


Intuitively, we allow the amplitude to be complex functions. That means, instead of constructing $\sqrt{f(p)}|0\rangle+\sqrt{1-f(p)}|1\rangle$, we try to construct $k_0(p)|0\rangle+k_1(p)|1\rangle$ where $k_i(p)$ might be a complex function.

\begin{definition}
Let $h(p)$ be a complex function of $p$. We define a quantum state $|f_h\rangle$ to be
$$|f_h\rangle=\frac{1}{\sqrt{1+|h(p)|^2}}(h(p)|0\rangle+|1\rangle).$$
\end{definition}
\begin{definition}
Let $k_0(p)$, $k_1(p)$ be two complex functions satisfying $|k_0(p)|^2+|k_1(p)|^2=1$. A tuple $(k_0(p),k_1(p))$ is said to be simulable if we can transform $|p\rangle$ to $|f_{\frac{k_0(p)}{k_1(p)}}\rangle=k_0(p)|0\rangle+k_1(p)|1\rangle$ by finite steps with a nonzero success probability, ignoring a global phase. In each step, we are allowed to use unitary transformation or measurement to the current state combined with some auxiliary qubits.

The auxiliary qubits can be other simulable states or some constant states $|a_c\rangle=\frac{1}{\sqrt{|a|^2+1}}(a|0\rangle+|1\rangle)$, where $a$ is a constant complex number which does not depend on $p$ or $f(\cdot)$.
\end{definition}

For example, $(\sqrt{f(p)},\sqrt{1-f(p)})$ is simulable when $f(p)=\frac{(1-2p)^2}{1+(1-2p)^2}$. We can implement $\sqrt{f(p)}|0\rangle + \sqrt{1-f(p)}|1\rangle$ by the following steps:
\begin{itemize}
\item[1] Use CNOT gate to $|\psi_0\rangle=|p\rangle|p\rangle$,we will get $|\psi_1\rangle=p|00\rangle+\sqrt{p(1-p)}|01\rangle+(1-p)|10\rangle+\sqrt{p(1-p)}|11\rangle$.
\item[2] Measure the second qubit of $|\psi_1\rangle$ in the computational basis. If we obtain the measurement outcome $|0\rangle$, the first qubit will collapse to $|\phi_0\rangle=\frac{1}{\sqrt{p^2+(1-p)^2}}(p|0\rangle+(1-p)|1\rangle$), and then to step 3. Otherwise, return to step 1.
\item[3] Use Hadamard gate to $|\phi_0\rangle$, and then use Pauli-X. We get the target state $|f(p)\rangle=\sqrt{\frac{(1-2p)^2}{1+(1-2p)^2}}|0\rangle+\sqrt{\frac{1}{1+(1-2p)^2}}|1\rangle$.
\end{itemize}


Given simulable $(k_0(p),k_1(p))$, it's not always obvious to know how to implement $|f_{\frac{k_0(p)}{k_1(p)}}\rangle$. However, in the following part, we will give an algorithm to implement them systematically.

\section{\label{sec:level1}Results OF Quantum-to-Quantum case and some discussion}

First, we give a sufficient and necessary condition to characterize the simulable tuple in the quantum-to-quantum case. The characterization of the real simulable tuple $(\sqrt{f(p)},\sqrt{1-f(p)})$ will be given later as a corollary. We use $\mathbb{C}$ to represent the complex number field, and $\mathbb{R}$ the real number field. The main theorem states the following.

\begin{theorem}\label{thm:field}
Let $k_0(p),k_1(p)$ be two complex functions satisfying $|k_0(p)|^2+|k_1(p)|^2=1$. Tuple $(k_0(p),k_1(p))$ is simulable if and only if $\frac{k_0(p)}{k_1(p)}$ belongs to the field generated by $\sqrt{\frac{p}{1-p}}$ and complex field $\mathbb{C}$. More precisely,
$$\frac{k_0(p)}{k_1(p)}=\frac{g_1(p)}{g_2(p)}\sqrt{\frac{p}{1-p}}+\frac{g_3(p)}{g_4(p)},$$
where $g_i(p)$ are polynomials of $p$ with coefficients in $\mathbb{C}$.
\end{theorem}

As we formally defined in Definition 1, tuple $(k_0(p),k_1(p))$  is simulable meaning that we can transform $|p\rangle$ to $|f_{\frac{k_0(p)}{k_1(p)}}\rangle$. To prove the theorem, the key point is to find a proper way to describe the change of amplitudes under quantum operations. For example, suppose we use a single qubit gate $U$ to $|p\rangle$ and obtain $U|p\rangle=s_0(p)|0\rangle+s_1(p)|1\rangle$. First, we know that $s_0(p)$ and $s_1(p)$ must be the linear combinations of $\sqrt{p}$ and $\sqrt{1-p}$. But be cautious, since $U$ is unitary, there are some constraints to $s_0(p)$ and $s_1(p)$. Even worse, when we are allowed to use measurement and auxiliary qubits, the amplitudes might become complicated and hard to describe.

So instead of considering amplitudes respectively, we consider the ratio $\frac{k_0(p)}{k_1(p)}$. Fortunately, we find that the set of $\frac{k_0(p)}{k_1(p)}$, where $(k_0(p),k_1(p))$ is simulable, has an intrinsic field structure: It is exactly the field generated by $\sqrt{\frac{p}{1-p}}$ and complex field $\mathbb{C}$.

 Let $S$ be the set of ratios of amplitudes in simulable tuple and $M$ be the set we mentioned in the theorem, that is,
  $$
  S=\{\frac{k_0(p)}{k_1(p)} \quad|\quad (k_0(p),k_1(p))\text{ is simulable} \}
  $$

$$
M=\{\frac{g_1(p)}{g_2(p)}\sqrt{\frac{p}{1-p}}+\frac{g_3(p)}{g_4(p)}\}
$$
where $g_i(p)$ is the polynomial of $p$ with coefficients in $\mathbb{C}$.

Our aim is to prove $S=M$.

\begin{proof}
\begin{sufficiency}

In this part, we want to prove $M \subseteq S$, which means if a $\frac{k_0(p)}{k_1(p)} \in M$, then $(k_0(p),k_1(p))$ is simulable. Here we give an intuitive proof to discover the structure of $S$, which implies $S$ is a field and has element $\sqrt{\frac{p}{1-p}}$. According to field theory, $M$ is exactly the field generated by $\mathbb{C}$ and $\sqrt{\frac{p}{1-p}}$. Thus we conclude $M\subseteq S$. Later, we provide an algorithm to implement $|f_{\frac{k_0(p)}{k_1(p)}}\rangle$.

Although we are allowed to use various kinds of unitary transformations, we show that a small set of unitary transformation is enough: We only need the Pauli-$X$ gate, the Hadamard gate, the controlled-NOT gate and a specific two-qubit gate (denote $B$ here).
\begin{equation*}       
X=\left(                 
  \begin{array}{cc}   
    0 &1\\  
   1 & 0 \\
  \end{array}
\right),
H=\frac{1}{\sqrt{2}}\left(                 
  \begin{array}{cc}   
    1&1\\  
   1 & -1 \\
  \end{array}
\right)   ,
\end{equation*}
\begin{equation*}       
CNOT=\left(                 
  \begin{array}{cccc}   
    1 & 0& 0&0\\  
   0&1 & 0&0\\
    0 & 0 &0&1\\
    0 & 0 &1&0\\
  \end{array}
\right),
B=\left(                 
  \begin{array}{cccc}   
    0& 0& 0&1\\  
   0&\frac{1}{\sqrt{2}} & \frac{1}{\sqrt{2}}&0\\
    0 & \frac{1}{\sqrt{2}} &\frac{-1}{\sqrt{2}}&0\\
    1 & 0 &0&0\\
  \end{array}
\right)   .             
\end{equation*}

Firstly, we prove $S$ is a field by showing $S$ is closed under addition and multiplication, and every element in $S$ has an multiplicative inverse. 

Suppose $h_1(p),h_2(p) \in S$, so we can implement $|\phi_i\rangle=c_i(p)(h_i(p)|0\rangle+|1\rangle) ( i=1,2)$, where $c_i(p)$ is used to normalize the vector.
\begin{itemize}
\item Inversion: Apply $X$ to $|\phi_1\rangle$, we get $c_1(p)(|0\rangle+h_1(p)|1\rangle)$. So multiplicative inverse $\frac{1}{h_1(p)} \in S$.
\item Multiplication: Apply $CNOT$ to $|\phi_1\rangle| \phi_2\rangle$, we get $CNOT|\phi_1\rangle| \phi_2\rangle=c_1(p)c_2(p)\{h_1(p)h_2(p)|00\rangle+h_1(p)|01\rangle+|10\rangle+h_2(p)|11\rangle\}$.
   Measure the second qubit in the computational basis. If we get the measurement outcome $|0\rangle$,
  the first qubit will collapse to $|\phi_3\rangle=c_3(p)(h_1(p)h_2(p)|0\rangle+|1\rangle)$, so $ h_1(p)h_2(p)\in S$. Specially, when $h_2(p)=c \in \mathbb{C}$,we know $ch_1(p) \in S$.
\item Addition: Apply $B$ to $|\phi_1\rangle| \phi_2\rangle$, we get $B|\phi_1\rangle| \phi_2\rangle=c_1(p)c_2(p)(|00\rangle+\frac{h_1(p)+h_2(p)}{\sqrt{2}}|01\rangle+\frac{h_1(p)-h_2(p)}{\sqrt{2}}|10\rangle+h_1(p)h_2(p)|11\rangle)$.
    Measure the first qubit in the computational basis. If we get measurement outcome $|0\rangle$,
    the second qubit will collapse to $|\phi_4\rangle=c_4(p)(|0\rangle+\frac{h_1(p)+h_2(p)}{\sqrt{2}}|1\rangle)$. Using the results in the previous two cases, we know $ h_1(p)+h_2(p) \in S$.
    \end{itemize}
So $S$ is a field. Initially, we have access to $|p\rangle$ and arbitrary constant qubits, so the generator of $S$ contains $\mathbb{C}$ and $\sqrt{\frac{p}{1-p}}$.

Secondly, use theory in field extension, we know $M$ is exactly the field generated by $\mathbb{C}$ and $\sqrt{\frac{p}{1-p}}$ . So we conclude $M\subseteq S$.
\end{sufficiency}

\begin{necessity}

In this part, we want to prove $S\subseteq M$: if $(k_0(p),k_1(p))$ is simulable, $\frac{k_0(p)}{k_1(p)}$ must belong to $M$. Let $n$ be the number of times we use the measurement operation in the simulation process. We will prove the necessity by using induction on $n$. Suppose we can implement $|\varphi\rangle=\sum_{j} s_{j}(p)|j\rangle$, note that $|\varphi\rangle$ may be a multi-qubit state. We state that

\begin{statement}
For any implementable $|\varphi\rangle =\sum_{j} s_{j}(p)|j\rangle$, the ratio of arbitrary two amplitudes of $|\varphi\rangle$ belongs to $M$, that is, $\forall k,l$, we have $\frac{s_{k}(p)}{s_{l}(p)} \in M$.
\end{statement}

The statement will lead to the proof of necessity.
If we reach a target qubit $|f_{\frac{k_0(p)}{k_1(p)}}\rangle$ by measuring other qubits of an implementable $|\varphi\rangle$,
The ratio of amplitudes of $|f_{\frac{k_0(p)}{k_1(p)}}\rangle$ will be equal to the ratio of two corresponding amplitudes of $|\varphi\rangle$. According to statement 1,we have $\frac{k_0(p)}{k_1(p)}\in M$.

Now, let us prove the statement. First, we show the statement keeps true under unitary transformation. In the simulation process after the $k$-th measurement operation, suppose we have already implemented several quantum states $|\varphi_i\rangle, i=1,...,s$. Without loss of generality, we assume $s=2$. When $s>2$, the analysis is similar.
    We write $|\varphi_i\rangle=\sum_{j} s_{ij}(p)|j\rangle (i=1,2)$. Then $|\varphi_1\rangle|\varphi_2\rangle=\sum_{k,r} s_{1k}(p)s_{2r}(p)|k\rangle\otimes|r\rangle$.

     If we use some unitary $U$ to $|\varphi_1\rangle|\varphi_2\rangle$ and get $|\psi\rangle=U|\varphi_1\rangle|\varphi_2\rangle$. The amplitude of $|\psi\rangle$ has the form
      $$\sum_{k,r} u_{kr}^{y} s_{1k}(p)s_{2r}(p)$$
     The ratio of two arbitrary amplitude of $|\psi\rangle$ equals to
      $$\frac{
      \sum_{k,r} u_{kr}^{y_1}
            {s_{1k}(p)/s_{10}(p)\cdot s_{2r}(p)/s_{20}(p)}
      }
      {
      \sum_{k,r} u_{kr}^{y_2}
            {s_{1k}(p)/s_{10}(p)\cdot s_{2r}(p)/s_{20}(p)}
      }$$

      According to the induction hypothesis, $s_{1k}(p)/s_{10}(p), s_{2r}(p)/s_{20}(p)\in M$. And for $M$ is a field, so the ratio of two arbitrary amplitude of $|\psi\rangle$ belongs to $M$.

Then we consider the measurement operation.

When $n=0$, we have $|p\rangle$ and arbitrary constant coins, the ratios of their amplitudes belong to $M$.

Suppose the statement is true for $n\leq k$. Suppose we have already implemented several quantum states $|\varphi_i\rangle, i=1,...,s$. Without loss of generality we assume $s=2$.
     And when $n=k+1$, if we measure some of the qubits of $U|\varphi_1\rangle|\varphi_2\rangle$ and keep the remaining qubit, the amplitudes of the remaining qubit will be some amplitudes of $U|\varphi_1\rangle|\varphi_2\rangle$, multiplying by a common factor to normalize the vector. The ratio of two amplitudes is still in $M$. As a conclusion, when $n=k+1$, Statement 1 is true.

    Thus, we complete our proof.
\end{necessity}
\end{proof}

Now, let us go back to the quantum coin $|f(p)\rangle=\sqrt{f(p)}|0\rangle+\sqrt{1-f(p)}|1\rangle$. We use our main theorem to characterize the simulable real tuple.
\begin{corollary}
Tuple $(\sqrt{f(p)},\sqrt{1-f(p)})$, where $f(p)$ is a real function in $[0,1]$, is simulable if and only if $\sqrt{\frac{f(p)}{1-f(p)}}$ belongs to the field generated by $\sqrt{\frac{p}{1-p}}$ and real field $\mathbb{R}$. More precisely,
$$\sqrt{\frac{f(p)}{1-f(p)}}=\frac{g_1(p)}{g_2(p)}\sqrt{\frac{p}{1-p}}+\frac{g_3(p)}{g_4(p)},$$
where $g_i(p)$ are polynomials of $p$ with coefficients in $\mathbb{R}$.
\end{corollary}
\begin{proof}
The sufficiency part is the same as  the proof of the main theorem. For the necessity part, first we use our theorem to show the ratio of amplitudes has a similar form as described in the corollary. Then notice that the ratio $\frac{g_1(p)}{g_2(p)}$, $\frac{g_3(p)}{g_4(p)}$ must be real, so it is easy to transform $g_i(p)$ to be a real function.
\end{proof}

\textbf{Algorithm to implement $|f_{\frac{k_0(p)}{k_1(p)}}\rangle$}

In this part, we briefly describe the algorithm to implement $|f_{\frac{k_0(p)}{k_1(p)}}\rangle$ based on $|p\rangle$ and constant states $|a_c\rangle$, when $(k_0(p),k_1(p))$ is  simulable.

 According to Theorem~\ref{thm:field}, $\frac{k_0(p)}{k_1(p)}$ has the decomposition
$$\frac{k_0(p)}{k_1(p)}=\frac{g_1(p)}{g_2(p)}\sqrt{\frac{p}{1-p}}+\frac{g_3(p)}{g_4(p)}$$
for some $g_i(p)=\sum_{j=0}^{n_i}a^{(ij)}p^j$, $a^{(ij)}\in \mathbb{C}$.

As proved in the theorem, we can use the addition and multiplication operations to transform $|f_{h_1}(p)\rangle,|f_{h_2}(p)\rangle$ to $|f_{h_1+h_2}(p)\rangle$ and $|f_{h_1h_2}(p)\rangle$.  Similarly, we can construct $|f_p(p)\rangle$ through Construct$_p$
$$|p\rangle\xrightarrow{\text{Multiply}  |p\rangle} |f_{\frac{p}{1-p}}\rangle\xrightarrow{\text{Add $|1_c\rangle$,then Inverse}}|f_{1-p}\rangle$$
$$\xrightarrow{\text{Add} |-1_c\rangle, \text{then Multiply} |-1_c\rangle}|f_p\rangle$$
Using functions Addition, Multiplication and Construct$_p$ we can implement $|f_{g_i}\rangle$ and finally $|f_{\frac{k_0(p)}{k_1(p)}}\rangle$.

\textbf{Some interesting functions in our simulable set}

%

In this part, we show some interesting examples and nonexamples of simuable functions.

First, for some simple real function $f(p)$, $(\sqrt{f(p)},\sqrt{1-f(p)})$ is not simulable in the quantum-to-quantum case.
For example, in the following.
\begin{statement}\label{sta:p2}
$(p,\sqrt{1-p^2})$ is not simulable.
\end{statement}
\begin{proof} We prove it by contradiction. Suppose $(p,\sqrt{1-p^2})$ is simulable. According to the theorem, it means there are $g_i(p)$ that satisfy
$$\frac{p}{\sqrt{1-p^2}}=\frac{g_1(p)}{g_2(p)}\sqrt{\frac{p}{1-p}}+\frac{g_3(p)}{g_4(p)},$$
where $g_i(p)$ are polynomials of $p$ with coefficients in $\mathbb{C}$, $g_1(p)\neq0, g_3(p)\neq0$. It implies
$$\frac{p^2}{1-p^2}=\frac{g_1^2(p)}{g_2^2(p)}\frac{p}{1-p}+\frac{g_3^2(p)}{g_4^2(p)}+2\frac{g_1(p)g_3(p)}{g_2(p)g_4(p)}\sqrt{\frac{p}{1-p}}$$
which means $\sqrt{\frac{p}{1-p}}$ equals to a rational function and thus leads to a contradiction.
\end{proof}

Be careful, although we can not construct $p|0\rangle+\sqrt{1-p^2}|1\rangle$, we can construct $e^{i\theta(p)}p|0\rangle+\sqrt{1-p^2}|1\rangle$ for some $\theta(p)$. The two states have the same probability of getting $|0\rangle$.
\begin{statement}\label{sta:p2again}
$(e^{i\theta(p)}p,\sqrt{1-p^2})$  is simulable, where $e^{i\theta(p)}=\frac{\sqrt{1-p^2}}{p}(\frac{\sqrt{2}p}{1+p}\sqrt{\frac{p}{1-p}}+\frac{pi}{1+p})$.
\end{statement}
\begin{proof}
Let $k_0(p)=e^{i\theta(p)}p, k_1(p)=\sqrt{1-p^2}$. Then $(k_0(p), k_1(p))$ is simulable for it satisfies the theorem.
\end{proof}


\textbf{Relations among simulable sets}

In this part, we will discuss the relations between our simulable sets and the simulable functions in the classical-to-classical case (referred to as CC) and quantum to classical case (referred to as QC).

First, in order to compare our simulable sets with CC and QC, we need to transform the simulable tuple $(k_0(p),k_1(p))$ into the simulable function. However, we need to carefully choose the way of transformation due to the subtle examples in statements~\ref{sta:p2} and \ref{sta:p2again}. Since our goal is to compare the simulable sets with the set CC and QQ in which a function is simulable means that a special probability event can be constructed, and the natural way to obtain probability from a quantum state is to measure the state, thus, we define the set of  simulable functions in the quantum-to-quantum case (referred as QQ) based on the probability of getting output $|0\rangle$.


 \begin{definition}The set of  simulable functions in the quantum-to-quantum case is defined to be
$$
QQ \triangleq \{|k_0(p)|^2, \text{where } (k_0(p),k_1(p)) \text{is simulable}\}
$$
\end{definition}
 Next, we will discuss the relationship among the three sets. Before that, first we list the main results of ~\cite{classical,quantum}.

\begin{definition}\cite{classical} A function $f(p)$ is polynomially bounded if there exists an integer $n$ such that
$$
\min(f(p),1-f(p))\geq \min(p^n,(1-p)^n)
$$
for $p \in [0,1]$.
\end{definition}

\begin{lemma}
(classical-to-classical~\cite{classical}) Let $f: [0,1] \rightarrow [0,1]$. The function $f$ is simulable, in other words $f\in CC$, if and only if

(1) $f$ is continuous on $[0,1]$;

(2) either $f$ is constant on $[0,1]$,  or $f$ is polynomially bounded.
\end{lemma}


The first observation is that many functions in QQ cannot be simulated in the classical-to-classical case.  For example, when
\begin{equation}
\frac{k_0(p)}{k_1(p)}=\prod_{i=1}^n (p-c_i), c_i \in (0,1).
\end{equation}
$(k_0(p),k_1(p))$ can be reached in the quantum-to-quantum case but $|k_0(p)|^2$  is not simulable in the classical case, for $|k_0(p)|^2$ reaches 0 in its domain~\cite{classical}.

On the other side, according to statement 3, we know $CC$ and $QQ$ have overlap $p^2$. Thus, we have
$$
CC\cap QQ\neq\emptyset,  QQ\not\subset CC.
$$

Next, we show that there is function $f(p)$ that belongs to $CC$ and does not belong to $QQ$. The function is shown in Fig.1.
\begin{equation*}
f(p)=\left\{
\begin{aligned}
\frac{1}{2},    &    p\in[0,\frac{1}{2}) \\
\frac{1}{2}p+\frac{1}{4},    &  p\in[\frac{1}{2},1]
\end{aligned}
\right.
\end{equation*}

\begin{figure}[H]
	\centering
		\centering\includegraphics[height=2in]{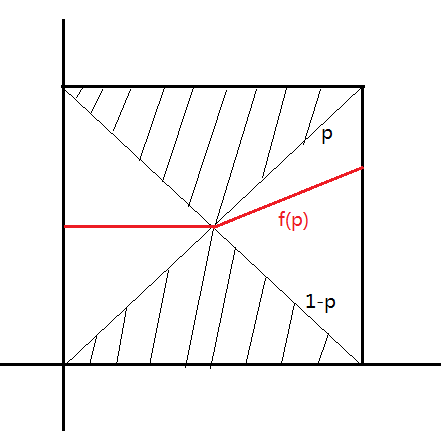}
		\caption{\label{fig:expo:random:a}CC$\not\subset QQ$}
\end{figure}

It's easy to show that $f(p)$  is polynomially bounded for $n=1$ and continuous in $[0,1]$, thus $f(p) \in CC$. We then show $f(p)\notin QQ$. If $f(p)\in QQ$, then $|k_0(p)|^2$ equals $\frac{1}{2}$ for $p\in[0,\frac{1}{2})$, then it must equal $\frac{1}{2} $ for $p\in[0,1]$, since the quantum-to-quantum simulable function or tuple has a strong relationship with polynomials and cannot have infinite zeros unless it is constant. So we conclude,
$$
CC\not\subset QQ.
$$

The following theorem proves $QC$ is strictly larger than $CC$,
$$
CC\subsetneqq QC.
$$
\begin{definition}\label{def:SPB}\cite{quantum} A function $f(p):[0,1] \rightarrow [0,1] $ is simple and poly-bounded(SPB) if and only if it satisfies

(1) $f$ is continuous.

(2) Both $Z=\{z_i: f(z_i)=0\}$ and $W=\{w_i:f(w_i)=1\}$ are finite sets.

(3) $\forall z\in Z$ there exist constants $c,\delta>0$ and integer $k <\infty$ such that
$$
c(p-z)^{2k}\leq f(p), \forall p\in [z-\delta,z+\delta].
$$

(4) $\forall w\in W$ there exist constants $c,\delta>0$ and integer $k <\infty$ such that
$$
1-c(p-w)^{2k}\geq f(p), \forall p\in [w-\delta,w+\delta].
$$
\end{definition}
\begin{lemma}(quantum-to-classical~\cite{quantum})
A function is simulable with quantum coins $|p\rangle=\sqrt{p}|0\rangle+\sqrt{1-p}|1\rangle$ and a finite set of single qubit unitary, in other words $f\in QC$, if and only if $f$ is SPB.
\end{lemma}

The quantum-to-classical case allows $f$ to reach 0 or 1 for finite times, and we conclude $CC\subsetneqq QC$.

Finally, we prove $QQ\subset QC$ by showing all functions in $QQ$ satisfy the $SPB$ condition. We put the proof in the appendix.

So we can conclude the three sets have the relationship shown in FIG.2.
\begin{figure}[H]
	\centering
		\centering\includegraphics[height=1in]{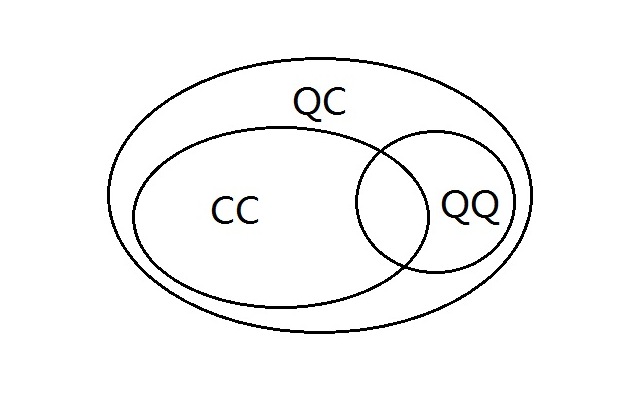}
		\caption{\label{fig:expo:random:a}Relation among three sets}
\end{figure}

\section{\label{sec:level1}Summary}
In our work, we give the complete answer to the quantum-input quantum-output Bernoulli factory problem and determine the exact set of simulable functions. The result provides a thorough understanding of the Bernoulli factory problem in the quantum world. The key point is the set of ratio $\frac{k_0(p)}{k_1(p)}$, where $(k_0(p),k_1(p))$ is simulable, and has a beautiful field structure. According to our main theorem, $p|0\rangle+\sqrt{1-p^2}|1\rangle$ can not be constructed while we can construct $e^{i\theta(p)}p|0\rangle+\sqrt{1-p^2}|1\rangle$ for some $\theta(p)$.  We also show the relationship between our results(QQ) and the quantum-to-classical case(QC) and the classical-to-classical case(CC). Note that CC and QQ don't have any inclusion relations with each other. Furthermore, for any simulable tuple $(k_0(p),k_1(p))$, we also give an algorithm to implement it.


\begin{appendix}
\section{\label{sec:level1}Proof of $QQ\subset QC$}

\begin{lemma} Let $T(x_1,x_2,x_3): \mathbb{R}^3\rightarrow\mathbb{R}$ be a multivariate polynomial of $x_1, x_2, x_3$. Suppose $T(p,\sqrt{p},\sqrt{1-p})$ is not a zero function. If $T(z,\sqrt{z},\sqrt{1-z})=0$ for some $z\in [0,1]$. Then there exists a real number $\delta$, a positive integer $k$, a function $m(p)$ which is continuous in $[z-\delta,z+\delta]$, such that $T(p,\sqrt{p},\sqrt{1-p})=(p-z)^{\frac{1}{2}k}m(p)$ and $m(z)\neq0$.
\end{lemma}
\begin{proof}
 For $T(x_1,x_2,x_3)$ are polynomials, there exist polynomials $t_k(p), k=1,2,3,4$, such that
\begin{displaymath}
\begin{array}{ll}
 & T(p,\sqrt{p},\sqrt{1-p})\\
=& t_1(p)+t_2(p)\sqrt{p}+t_3(p)\sqrt{1-p}+t_4(p)\sqrt{p(1-p)}
\end{array}
\end{displaymath}

\begin{itemize}
\item[(a)] When $z\in(0,1)$.
If $(p-z)| t_k(p), k=1,2,3,4$, then we can extract $(p-z)$. If not, we assume $t_1(z)\neq 0$, other situation is similar. Furthermore, we can assume $t_1(z)+t_2(z)\sqrt{z}\neq 0$, for if every two items sum to 0, combined with $T(z,\sqrt{z},\sqrt{(1-z)}=0$, we can conclude $t_1(z)=0$. Then $-t_1(z)-t_2(z)\sqrt{z}+t_3(z)\sqrt{1-z}+t_4(z)\sqrt{z(1-z)}\neq 0$ and we can rewrite T as
$$
\frac{
(t_3(p)\sqrt{1-p}+t_4(p)\sqrt{p(1-p)})^2-(t_1(p)+t_2(p)\sqrt{p})^2
}
{
-t_1(p)-t_2(p)\sqrt{p}+t_3(p)\sqrt{1-p}+t_4(p)\sqrt{p(1-p)}
}
$$
$z$ is zero of the numerator, but not a zero of the denominator. Thus we can let the denominator be a part of $m(p)$. Now, since we can eliminate $\sqrt{1-p}$ in the numerator, we can rearrange it as $t'_1(p)+t'_2(p)\sqrt{p}$. Next, we use the similar technique. If $z$ is zero of both $t'_1, t'_2$, extract $(p-z)$; otherwise, $t'_1(z)-t'_2(z)\sqrt{z}\neq 0$, and we have the numerator $=\frac{t''(p)}{t'_1(p)-t'_2(p)\sqrt{p}}$ for some polynomial $t''(p)$. Finally, $z$ is zero of $t''(p)$ so that there exists integer $\ell$ and polynomial $s(p)$ such that $t''(p) = (p-z)^{\ell}s(p)$ and $s(z)\neq 0$.  Combine all of the argument, we prove the lemma in this case.
\item[(b)] When $z=0$, then $t_1(0)+t_3(0)=0$.
If $t_1(0)=t_3(0)=0$, thus $p|t_1(p), p|t_3(p)$. So we can extract an common divisor $\sqrt{p}$ and consider $\frac{T(p,\sqrt{p},\sqrt{1-p})}{\sqrt{p}}$ instead. If $t_1(0)\neq0$, then $t_1(0)+t_2(0)\sqrt{0}\neq0$. We can use techniques in situation(a) to simplify the numerator. The following steps are similar.
\item[(c)] When $z=1$, techniques are similar to case (b).

\end{itemize}
Combine (a)(b)(c), we complete our proof.
\end{proof}

\begin{statement}
If $f(p)\in QQ$, then $f(p)$ satisfies the SPB condition (see Definition~\ref{def:SPB}). Thus $QQ\subset QC$.
\end{statement}
\begin{proof}
If $f(p)\in QQ$, then there exists a simulable tuple $(k_0(p),k_1(p))$ such that $f(p)=|k_0(p)|^2$. According to the theorem, there exist the complex multivariate polynomials $T_1(x_1,x_2,x_3), T_2(x_1,x_2,x_3)$ such that
$$
\frac{k_0(p)}{k_1(p)}=\frac{T_1(p,\sqrt{p},\sqrt{1-p})}{T_2(p,\sqrt{p},\sqrt{1-p})}
$$
For convenience, we use $T_1=T_1(p,\sqrt{p},\sqrt{1-p})$, $T_2=T_2(p,\sqrt{p},\sqrt{1-p})$.
Then
$$
|k_0(p)|^2=\frac{|T_1|^2}{|T_1|^2+|T_2|^2}.
$$
In the following, we check three conditions in the SPB definition.
\begin{itemize}
\item[(1)] Both $Z=\{z_i: f(z_i)=0\}$ and $W=\{w_i:f(w_i)=1\}$ are finite sets.

For $T_i(x_1,x_2,x_3)$ are multivariate polynomials so $T_i$ is bounded when $p \in [0,1]$. Zeros of $f(p)$ should be zeros of $T_1$. It is easy to show that solving $T_1(p,\sqrt{p},\sqrt{1-p})=0$ can be transformed to find zeros of an univariate polynomial. So $f(p)$ has only finite zeros. Similarly, consider zeros of $f_1(p)=|k_1(p)|^2$, we know $f(p)$ has finite ones.

\item[(2)] $f(p)$ is continuous in $[0,1]$.

  For $T_i$ has finite zeros and continuous in $[0,1]$, to prove the continuity of $f$, we need to consider the value of $f(p)$ when $|\frac{T_2}{T_1}|^2\rightarrow\infty$. At this time, $f(p)$ is bounded and goes to 0, so $f(p)$ is continuous in [0,1].

\item[(3)]$\forall z\in Z$ there exists constants $c,\delta>0$ and integer $k <\infty$ such that
$$
c(p-z)^{2k}\leq f(p), \forall p\in [z-\delta,z+\delta].
$$

Use lemma 3, we can easily get this conclusion. Similarly, consider zeros of $|k_1(p)|^2$, we can get the part for ones of $f(p)$.
\end{itemize}
\end{proof}
\end{appendix}



\end{document}